\documentclass[onecolumn]{IEEEtran}
\usepackage[utf8]{inputenc} 
\usepackage[T1]{fontenc} % use cm-super to avoid Type 3 fonts
\usepackage{amsthm}
\usepackage[cmex10]{amsmath}
\usepackage{bm}
\usepackage{times,amssymb,amsfonts,float,nicefrac,color,mathrsfs,caption} % newtxmath uses less space; can use vmathbb; remove bbm
\usepackage{enumerate,multirow,caption,tikz,graphicx,enumitem,soul}
\usepackage[mathscr]{eucal}
\usepackage{sidecap}
\usepackage{algpseudocode}
\usepackage{verbatim}
\usepackage{epstopdf}
\usepackage{textcomp}
\usetikzlibrary{shapes,arrows}
\usepackage[noadjust]{cite}
\usepackage{booktabs}
\allowdisplaybreaks

\usepackage[font={it}]{caption}
\usepackage[margin=0.05in]{subcaption}
\usepackage[bottom]{footmisc}
\usepackage[normalem]{ulem}

\usepackage[ruled,vlined,linesnumbered]{algorithm2e}

\usepackage{hyperref} % ref online resource
\usepackage{stmaryrd}
\usepackage{balance}

\newtheoremstyle{mytheoremstyle}{3pt}{3pt}{\itshape}{}{\bf}{.}{.3em}{} 
\theoremstyle{mytheoremstyle}
\newtheorem{theorem}{Theorem}

\newcommand\nc\newcommand
\nc\bfa{{\boldsymbol a}}\nc\bfA{{\boldsymbol A}}\nc\cA{{\mathscr A}}
\nc\bfb{{\boldsymbol b}}\nc\bfB{{\boldsymbol B}}\nc\cB{{\mathscr B}}
\nc\bfc{{\boldsymbol c}}\nc\bfC{{\boldsymbol C}}\nc\cC{{\mathscr C}}
\nc\bfd{{\boldsymbol d}}\nc\bfD{{\boldsymbol D}}\nc\cD{{\mathscr D}}
\nc\bfe{{\boldsymbol e}}\nc\bfE{{\boldsymbol E}}\nc\cE{{\mathscr E}}
\nc\bff{{\boldsymbol f}}\nc\bfF{{\boldsymbol F}}\nc\cF{{\mathscr F}}
\nc\bfg{{\boldsymbol g}}\nc\bfG{{\boldsymbol G}}\nc\cG{{\mathscr G}}
\nc\bfh{{\boldsymbol h}}\nc\bfH{{\boldsymbol H}}\nc\cH{{\mathscr H}}
\nc\bfi{{\boldsymbol i}}\nc\bfI{{\boldsymbol I}}\nc\cI{{\mathcal I}}
\nc\bfj{{\boldsymbol j}}\nc\bfJ{{\boldsymbol J}}\nc\cJ{{\mathscr J}}
\nc\bfk{{\boldsymbol k}}\nc\bfK{{\boldsymbol K}}\nc\cK{{\mathscr K}}
\nc\bfl{{\boldsymbol l}}\nc\bfL{{\boldsymbol L}}\nc\cL{{\mathscr L}}
\nc\bfm{{\boldsymbol m}}\nc\bfM{{\boldsymbol M}}\nc\cM{{\mathscr M}}\nc\fM{{\mathfrak M}}
\nc\bfn{{\boldsymbol n}}\nc\bfN{{\boldsymbol N}}\nc\cN{{\mathscr N}}
\nc\bfo{{\boldsymbol o}}\nc\bfO{{\boldsymbol O}}\nc\cO{{\mathscr O}}
\nc\bfp{{\boldsymbol p}}\nc\bfP{{\boldsymbol P}}\nc\cP{{\mathscr P}}\nc\eP{{\EuScriptP}}\nc\fP{{\mathfrak P}}
\nc\bfq{{\boldsymbol q}}\nc\bfQ{{\boldsymbol Q}}\nc\cQ{{\mathscr Q}}
\nc\bfr{{\boldsymbol r}}\nc\bfR{{\boldsymbol R}}\nc\cR{{\mathscr R}}
\nc\bfs{{\boldsymbol s}}\nc\bfS{{\boldsymbol S}}\nc\cS{{\mathscr S}}
\nc\bft{{\boldsymbol t}}\nc\bfT{{\boldsymbol T}}\nc\cT{{\mathscr T}}
\nc\bfu{{\boldsymbol u}}\nc\bfU{{\boldsymbol U}}\nc\cU{{\mathscr U}}
\nc\bfv{{\boldsymbol v}}\nc\bfV{{\boldsymbol V}}\nc\cV{{\mathscr V}}
\nc\bfw{{\boldsymbol w}}\nc\bfW{{\boldsymbol W}}\nc\cW{{\mathscr W}}
\nc\bfx{{\boldsymbol x}}\nc\bfX{{\boldsymbol X}}\nc\cX{{\mathscr X}}
\nc\bfy{{\boldsymbol y}}\nc\bfY{{\boldsymbol Y}}\nc\cY{{\mathscr Y}}
\nc\bfz{{\boldsymbol z}}\nc\bfZ{{\boldsymbol Z}}\nc\cZ{{\mathscr Z}}

\newtheorem{lemma}[theorem]{Lemma}

\newtheorem{proposition}[theorem]{Proposition}

\newtheorem{definition}{Definition}

\newtheorem{construction}{Construction}[section]
\newtheorem{example}{Example}[section]

\theoremstyle{remark}
\newtheorem{remark}{Remark}[section]

\newcommand{\etal}{{\em et al.\ }}

\newcommand{\ff}{{\mathbb F}}
\DeclareMathOperator{\wt}{wt}

\newcommand{\cev}[1]{\overset{{}_{\shortleftarrow}}{#1}}

\begin{document}
	
	\title{Coding methods for string reconstruction from erroneous prefix-suffix compositions}
%	\author{%
%		\IEEEauthorblockN{Zitan Chen}
%		\IEEEauthorblockA{
%			School of Science and Engineering\\
%			Future Networks of Intelligence Institute\\
%			The Chinese University of Hong Kong, Shenzhen\\
%			Email: chenztan@cuhk.edu.cn}
%	}
	
	\author{
%		\IEEEauthorblockN{} \hspace*{1in}
%		\and
		\IEEEauthorblockN{Zitan Chen}
		}
	\maketitle	
	
	{\renewcommand{\thefootnote}{}\footnotetext{
			
			\vspace{-.2in}
			
			\noindent\rule{1.5in}{.4pt}

			{	
				Z. Chen is with School of Science and Engineering, Future Networks of Intelligence Institute, The Chinese University of Hong Kong, Shenzhen, China (Email: chenztan@cuhk.edu.cn).
			}
		}
	}
	\renewcommand{\thefootnote}{\arabic{footnote}}
	\setcounter{footnote}{0}

	\begin{abstract} 
		The number of zeros and the number of ones in a binary string are referred to as the composition of the string, and the prefix-suffix compositions of a string are a multiset formed by the compositions of the prefixes and suffixes of all possible lengths of the string. In this work, we present binary codes of length $n$ in which every codeword can be efficiently reconstructed from its erroneous prefix-suffix compositions with at most $t$ composition errors. All our constructions have decoding complexity polynomial in $n$ and the best of our constructions has constant rate and can correct $t=\Theta(n)$ errors. As a comparison, no prior constructions can afford to efficiently correct $t=\Theta(n)$ arbitrary composition errors.
		
		Additionally, we propose a method of encoding $h$ arbitrary strings of the same length so that they can be reconstructed from the multiset union of their error-free prefix-suffix compositions, at the expense of $h$-fold coding overhead. In contrast, existing methods can only recover $h$ distinct strings, albeit with code rate asymptotically equal to $1/h$. Building on the top of the proposed method, we also present a coding scheme that enables efficient recovery of $h$ strings from their erroneous prefix-suffix compositions with $t=\Theta(n)$ errors.

	\end{abstract}
	
	%	{\IEEEkeywords }
	
	\section{Introduction}
	Recent advancement in macromolecular technology has brought polymer-based storage systems \cite{al2017mass,launay2021precise} close to reality. However, one of the major problems that must be addressed in order to make such systems practical is the problem of data retrieval based on sequencing, or rather polymer reconstruction from their partial information. 
	
	Motivated by mass spectrometry sequencing \cite{chen2020bioinformatics}, under proper assumptions, one may represent polymers by binary strings and the problem of polymer reconstruction can be turned into string reconstruction from {substring} compositions, i.e., from the number of zeros and the number of ones in substrings of every possible length. 
	Note that a string and its reversal have the same substring compositions and thus are indistinguishable. The pioneering work of \cite{acharya2015string} characterized the length for which strings can be uniquely reconstructed from their substring compositions up to reversal.
	Extending the work of \cite{acharya2015string}, the authors of \cite{pattabiraman2023coding} and \cite{banerjee2023insertion} studied the problem of string reconstruction from {erroneous} substring compositions. Specifically, Pattabiraman \etal \cite{pattabiraman2023coding} designed coding schemes capable of reconstructing strings in the presence of {substitution} errors, and Banerjee \etal \cite{banerjee2023insertion} further proposed codes that can deal with {insertion} and {deletion} errors.
	
	In addition to identifying strings that allow unique reconstruction and error correction, low complexity algorithms for string reconstruction are of great interest, as efficiency of data retrieval is crucial for practical polymer-based storage systems. In the case of reconstruction from error-free substring compositions, the work of \cite{acharya2015string} described a backtracking algorithm for strings of length $n$ with worst-case time complexity exponential in $\sqrt{n}$. Moreover, the subsequent works \cite{pattabiraman2023coding,gupta2025new} constructed sets of strings that can be uniquely reconstructed with time complexity polynomial in $n$.
	For reconstruction in the presence of substitution composition errors, the work \cite{pattabiraman2023coding} showed that when the number of errors is a constant independent of $n$, there exist coding schemes with decoding complexity polynomial in $n$.

	Observing that it may not be realistic to assume the compositions of all substrings are available, the authors of \cite{gabrys2022reconstruction} initiated the study of string reconstruction based on the compositions of {prefixes and suffixes} of all possible lengths. In fact, the authors of \cite{gabrys2022reconstruction} considered the more general problem of reconstructing {multiple} distinct strings of the same length simultaneously from the compositions of their prefixes and suffixes. The main result of \cite{gabrys2022reconstruction} reveals that for reconstruction of no more than $h$ distinct strings of the same length, there exists a family of codes with rate approaching $1/h$ asymptotically. The idea that is instrumental in constructing such codes is the use of binary $B_h$-sequences, i.e., a set of strings with the property that for any two distinct subsets of strings of size at most $h$, the real-valued sums of all the strings in each of the two subsets are different.
	Following the work \cite{gabrys2022reconstruction}, the authors of \cite{ye2024reconstruction} studied in depth the problem of reconstructing a single string from the compositions of its prefixes and suffixes. In particular, their work completely characterized the strings that can be reconstructed from its prefixes and suffixes compositions uniquely up to reversal.
	Both works \cite{gabrys2022reconstruction,ye2024reconstruction} also presented sets of strings that can be efficiently reconstructed. 
	More recently, the work of \cite{yang2025reconstruction} considered the problem of reconstructing $h$ strings that are not necessarily distinct but have the same weight from their error-free compositions of prefixes and suffixes, and derived necessary and sufficient conditions for which unique reconstruction up to the reversal is possible. For reconstruction of multiple distinct strings of constant weight, we note that the recent work \cite{sima2023constant} studied the largest possible set of constant-weight binary $B_2$-sequences, which can be used in the framework of \cite{gabrys2022reconstruction} to construct constant-weight codes in which any two distinct strings can be reconstructed uniquely from the compositions of their prefixes and suffixes.
	
	\subsection{Our results}
	In this paper, we consider the problem of string reconstruction from erroneous prefix-suffix compositions resulted from general composition errors that include any combination of insertion, deletion and substitution of compositions, and present construction of codes that allow efficient reconstruction from erroneous prefix-suffix compositions. We note that the prior works \cite{gabrys2022reconstruction,ye2024reconstruction} also studied string reconstruction from erroneous prefix-suffix compositions and presented coding schemes capable of correcting certain composition errors. However, the coding schemes proposed in these works can only cope with limited types of errors. Specifically, the work \cite{gabrys2022reconstruction} discussed how to reconstruct strings even if there are missing compositions (i.e., deletion errors) in the prefix-suffix compositions, and the work \cite{ye2024reconstruction} further studied codes that are able to handle a special type of substitution errors caused by a reduced number of ones in some compositions.

	Using generalized Reed-Solomon (GRS) codes as the main building block, we first present a basic construction of linear codes of length $n$ over a prime field. Developing from this construction, we further propose several binary codes of length $n$ that are capable of correcting at most $t$ composition errors. The decoding of these codes can be accomplished by variants of the decoders for GRS codes, and thus all of them have decoding complexity polynomial in $n$. Moreover, utilizing asymptotically good binary codes that are efficiently constructable and decodable, we present a construction of binary codes with constant rate that can correct $t=\Theta(n)$ composition errors and recover the underlying string efficiently.
	
	Going beyond reconstruction of a single string, we also study the problem of reconstructing $h>1$ strings from the multiset of their prefix-suffix compositions. Specifically, we describe a method of jointly coding $h$ arbitrary strings of the same length so that they can be efficiently reconstructed from their error-free prefix-suffix composition. The rate of the resulting code is $1/(h+1)$. As noted before, the work of \cite{gabrys2022reconstruction} presented a family of codes with rate approaching $1/h$ asymptotically. However, the codes in \cite{gabrys2022reconstruction} can only afford reconstruction of $h$ distinct strings. Finally, we extend this coding method by incorporating it with asymptotically good binary codes to enable reconstruction in the presence of composition errors, giving rise to a construction of codes with constant rate that can efficiently correct $t=\Theta(n)$ composition errors. 
	
	\subsection{Organization}
	The rest of this paper is organized as follows. 
In Section~\ref{sec:pre}, we introduce our notation, describe the problem setting, and review some known results in coding theory that are useful for subsequent discussion. Section~\ref{sec:con} is devoted to code constructions. Specifically, coding schemes for reconstructing a single string from its erroneous prefix-suffix compositions is presented in Section~\ref{sec:con-single} while coding schemes for reconstructing multiple strings are present in Section~\ref{sec:con-multiple}. Finally, we conclude this paper and mention a few problems for further investigation in Section~\ref{sec:re}.
	
	\section{Preliminaries}\label{sec:pre}
	
	Let $n$ be a positive integer. 
	Let $\bm{c}=c_1c_2\ldots c_n \in \{0,1\}^n$ be a binary string of length $n$, and the {\emph{reversal}} of $\bm{c}$ is denoted by $\cev{\bm{c}}:=c_nc_{n-1}\ldots c_1$. The {\emph{weight}} of $\bm{c}$ is the number of ones in $\bm{c}$, denoted by $\wt(\bm{c})$. The {\emph{composition}} of $\bm{c}$ is formed by the number of zeros and the number of ones in $\bm{c}$. More precisely, the ordered pair $(n-\wt(\bm{c}),\wt(\bm{c}))$ is called the composition of $\bm{c}$, where $n$ and $\wt(\bm{c})$ are referred to as the \emph{size} and \emph{mass} of the composition, respectively.
	For $1\leq l\leq n$, the length-$l$ prefix and the length-$l$ suffix of $\bm{c}$ are denoted by $\bm{c}[l]$ and $\cev{\bm{c}}[l]$, respectively. 
	{We	will use ``$\cup$'' to denote both the set union and the multiset union. The exact meaning of ``$\cup$'' will be clear from the context.}
	
	\begin{definition}\label{def:M}
		The {set} of compositions of all prefixes of a string $\bm{c}\in\{0,1\}^n$ is called the prefix {compositions} of $\bm{c}$, denoted by $P(\bm{c})$. More precisely, 
		\begin{align*}
			P(\bm{c})=\{ (j-\wt(\bm{c}[j]),\wt(\bm{c}[j])) \mid 1\leq j\leq n\}.
		\end{align*}
		The suffix {compositions} of $\bm{c}$ are {similarly defined to be}
		\begin{align*}
			P(\cev{\bm{c}})=\{ (j-\wt(\cev{\bm{c}}[j]),\wt(\cev{\bm{c}}[j])) \mid 1\leq j\leq n\}.
		\end{align*} 
		The prefix-suffix compositions of $\bm{c}$ are defined to be the multiset union of $P(\bm{c})$ and $P(\cev{\bm{c}})$, denoted by $M(\bm{c})$. 
	\end{definition}
	
	Let $\fM_n$ be the collection of multisets of finite cardinality in which each element is an ordered pair $(a,b)\in \{0,1,\ldots,n\}^2$ such that $1\leq a+b\leq n$. Clearly, $M(\bm{c})\in\fM_n$. For each multiset $X\in \fM_n$, we will also write $X$ as an $n$-tuple $(X_j)_{1\leq j\leq n}$ where $X_j:=\{(a,b)\in X\mid a+b = j\}$ is a multiset. For $X,Y\in \fM_n$, the \emph{distance} between $X,Y$ is defined to be
	\begin{align*}
		d(X,Y):=|\{j\mid X_j\neq Y_j\}|.
	\end{align*}
	For two sequences of numbers $\bm{x}=(x_j)_{1\leq j\leq n},\bm{y}=(y_j)_{1\leq j\leq n}$, we also use $d(\bm{x},\bm{y})$ to denote the \emph{Hamming distance} between $\bm{x}$ and $\bm{y}$.
	
	Let $Y\in\fM_n$, $(a,b)\in M(\bm{c})$, and $(a',b')\in \{0,1,\ldots,n\}^2$ be such that $1\leq a'+b'\leq n$. We say that $Y$ has a composition insertion error with respect to $M(\bm{c})$ if $Y$ is the multiset union $M(\bm{c})\cup \{(a',b')\}$; $Y$ has a composition deletion error if $Y=M(\bm{c})\setminus\{(a,b)\}$; $Y$ has a composition substitution error if $Y=(M(\bm{c})\setminus\{(a,b)\})\cup\{(a',b')\}$.
	In all these cases, $Y$ is said to have a single composition error and we have $d(M(\bm{c}),Y)=1$.
	More generally, let $t$ be a positive integer. We say that $Y$ is a multiset of erroneous prefix-suffix compositions resulting from $t$ errors in $M(\bm{c})$ if $d(M(\bm{c}),Y)= t$. Note that this definition takes into account any combination of insertion, deletion and substitution composition errors. Moreover, under this definition, inserting (or deleting) multiple compositions of the same size are counted as one error.
	
	As observed in the earlier works such as \cite{acharya2015string,pattabiraman2023coding}, string reconstruction from compositions becomes more manageable if the strings satisfy the property that the weight of every prefix of length up to half the string is smaller (or larger) than the weight of the suffix of the same length. In particular, the work of \cite{pattabiraman2023coding} presented a construction of codes with this property, stated below.
	
	\begin{lemma}[\cite{pattabiraman2023coding}]\label{le:ordered}
		There exists a code $\cS_n\subset\{0,1\}^n$ with $n \geq 8$ and redundancy $O(\log n)$ such that for any $\bm{c}\in\cS_n$ it holds that $\wt(\bm{c}[j])\leq \wt(\cev{\bm{c}}[j])$ for all $j\leq n/2$. Moreover, the code can be encoded and decoded in time polynomial in $n$.
	\end{lemma}
	
	We note that if the weight condition $\wt(\bm{c}[j])\leq \wt(\cev{\bm{c}}[j])$ holds for every prefix of length up to half the string, i.e., $j\leq n/2$, then it is also satisfied by the prefixes of all possible lengths, as shown in the following proposition. This property is helpful in distinguishing prefix weights from suffix weights, and thus is instrumental in string reconstruction from prefix-suffix compositions.
	
	\begin{proposition}\label{prop:weight}
		Let $\bm{c}\in\{0,1\}^n$. If $\wt(\bm{c}[j])\leq \wt(\cev{\bm{c}}[j])$ for all $j\leq n/2$, then it holds that $\wt(\bm{c}[j])\leq \wt(\cev{\bm{c}}[j])$ for all $j\leq n$.
	\end{proposition}
	\begin{proof}
		For $j=1,\ldots,n$, let $s_j=\wt(\bm{c}[j])$ and $\bar{s}_j=\wt(\cev{\bm{c}}[j])$. 
		Define $s_0=\bar{s}_0=0$. Then we have $\bar{s}_j=s_n-s_{n-j}$ and $s_j=s_n-\bar{s}_{n-j}$ for $j=0,\ldots,n$.
		Since $s_j\leq \bar{s}_j$ for all $j\leq n/2$, we have $s_n-\bar{s}_{n-j}\leq s_n-s_{n-j}$ for all $j\leq n/2$.  
		It follows that $s_j\leq \bar{s}_j$ for all $j\leq n$. 
	\end{proof}
	
	We next mention a few basic results on error-correcting codes. Let $\ff_q$ be a finite field of order $q$ where $q$ is a prime power denote $\ff_q^*=\ff_q\setminus\{0\}$. Let $\{\alpha_1,\ldots,\alpha_{n}\}\subset \ff_q$ and $(\omega_1,\ldots,\omega_n)\in(\ff_q^*)^n$. Let $k\leq n$ be a positive integer. An $(n,k)$ GRS code over $\ff_q$ is a linear code over $\ff_q$ with a parity check matrix given by 
	\begin{align*}
		\begin{pmatrix}
			\omega_1 & \dots & \omega_n\\
			\omega_1\alpha_1 & \dots & \omega_n\alpha_n\\
			\vdots & \ddots & \vdots\\
			\omega_1\alpha_1^{n-k-1} & \dots & \omega_n\alpha_n^{n-k-1}
		\end{pmatrix}.
	\end{align*}
	The elements $\alpha_1,\ldots,\alpha_n$ are called evaluation points and the elements $\omega_1,\ldots,\omega_n$ are called column multipliers. GRS codes are known to has minimum distance $n-k+1$ and there are many decoding algorithms for correcting $t$ errors with time complexity polynomial in $n$ if $t$ is smaller than half of the minimum distance. We will use GRS codes as building blocks for some of our code constructions later.
	
	Another class of codes we will use is the binary BCH codes. They are also known to have polynomial-time decoding algorithms that correct $t$ errors if $t$ is smaller than half of the designed distance. We will need the following result for our code construction later.
	
	\begin{lemma}[\cite{macwilliams1977theory}]\label{le:bch}
		Let $m$ be a positive integer and $2t-1\leq 2^{\lceil m/2\rceil}+1$.
		There exists a BCH code $\cB_n\subset\{0,1\}^n$ of length $n=2^m-1$, dimension $n-mt$, and designed distance $2t+1$.
	\end{lemma}
	
	We will also be using asymptotically good codes over the binary field $\ff_2$. A family of $(n,k)$ linear codes over $\ff_q$ with minimum distance $d$ is said to be asymptotically good if the rate $R:=k/n$ and the growth rate of the minimum distance relative to the block length, i.e., the relative distance $d/n$, are non-vanishing as $n$ goes to infinity. In other words, the rate and relative distance of asymptotically good codes are both positive constants independent of the block length. Examples of explicit construction of asymptotically good codes over the binary field include expander codes \cite{sipser1996expander} and concatenated codes \cite{forney1965concatenated}. These codes are known to have polynomial-time decoding algorithms that correct a number of errors that is a constant fraction of the block length \cite{zemor2001expander}, \cite{forney1966generalized}. For the purposes of this paper, we will simply assume the availability of explicit asymptotically good binary codes with efficient decoding algorithms. For the ease of reference, let us summarize these results in the following lemma.
	
	\begin{lemma}\label{le:good}
		There exists an explicit family of linear codes $\cA_n \subset \{0,1\}^n$ with constant rate that can correct $t=\Theta(n)$ errors in time polynomial in $n$. 
	\end{lemma}
		
	\section{Code constructions}\label{sec:con}
	
	\subsection{Reconstruction of a single string}\label{sec:con-single}
	In this subsection we present coding schemes for reconstructing a single string from its erroneous prefix-suffix compositions. To begin with, let us present a construction of linear codes over a prime field from which binary codes capable of correcting $t$ composition errors can be derived. The basic idea underlying this construction is to encode vectors over a large enough prime field so that the prefix sums of each vector form a codeword in a GRS code. As we are interested in reconstruction of binary strings, let us consider the binary vectors of this GRS code. Given the prefix-suffix compositions, i.e., the prefix sums and suffix sums, of the binary vectors, if one can further distinguish the prefix sums from the suffix sums, then the error-correcting capability of the GRS code will allow one to recover the prefix sums even if there are composition errors. With the correct values of the prefix sums, one can then easily deduce the corresponding binary vector.
	
	\begin{construction}\label{con:basic}
		Let $t< n/2$ be a positive integer. Let $\ff_p$ be a finite field of order $p$ where $p\geq n+1$ is a prime. Let $\{\alpha_1,\dots,\alpha_n\}\subset\ff_p$ and $(\omega_j)_{1\leq j\leq n}\in (\ff_p^*)^n$. 
		Let $\cD_n:=\{(c_j)_{1\leq j\leq n}\}$ be the linear code over $\ff_p$ that satisfies the following parity-check equations: 
		\begin{align}
			\sum_{j=1}^{n}\sum_{i=j}^{n}\omega_i\alpha_i^l c_j=0,\quad l=0,\ldots,2t-1.\label{eq:pc}
		\end{align}
	\end{construction}
	
	Below we consider the largest binary subcode contained in $\cD_n$. There is no efficient description of this subcode as one may have to enumerate every binary string of length $n$ in order to determine its membership in $\cD_n$. The next lemma shows if the prefix weights of the strings in this subcode can be distinguished from the suffix weights, then it is possible to correct composition errors and recover the underlying string.
	
	\begin{lemma}\label{le:basic}
		Let $\bm{c}\in \cD_n\cap\{0,1\}^n$ where $\cD_n$ is given in Construction~\ref{con:basic}. Let $Y\in\fM_n$ be any erroneous prefix-suffix compositions resulting from at most $t$ errors in $M(\bm{c})$. If $\wt(\bm{c}[j])\leq \wt(\cev{\bm{c}}[j])$ for all $j\leq n/2$, then one can recover $\bm{c}$ from $Y$.
	\end{lemma}
	
	\begin{proof}
		Denote $X=M(\bm{c})\in\fM_n$ and note that $d(X,Y)\leq t$. Without loss of generality, we may assume $|Y_j|=2$ for $j=1,\ldots,n$. Indeed, we may discard $|Y_j|-2$ arbitrary instances in $Y_j$ for all $j$ such that $|Y_j|>2$ to obtain another $Y'\in \fM_n$ satisfying $d(X,Y')\leq t$. If $|Y_j|<2$, then adding $2-|Y_j|$ instances of $(j,0)$ produces another $Y''\in \fM_n$ satisfying $d(X,Y'')\leq t$.
		 
		For $j=1,\ldots,n$, let $s_j=\wt(\bm{c}[j])$ and $\bar{s}_j=\wt(\cev{\bm{c}}[j])$. 
		Since $\bm{c}\in\{0,1\}^n$, we have $s_j$ is the prefix sum of $\bm{c}=c_1c_2\ldots c_n$ of length $j$, and $\bar{s}_j$ is the corresponding suffix sum.
		Since $s_j\leq \bar{s}_j$ for all $j\leq n/2$, by Proposition~\ref{prop:weight} we have $s_j\leq \bar{s}_j$ for all $j\leq n$. 
		Thus, one can distinguish the prefix sum of length $j$ from the suffix sum of the same length based on their values. 
		 
		Define $\bm{x}=(s_j)_{1\leq j\leq n}$. 
		Moreover, let us write $Y_j=\{(a_j,b_j),(\bar{a}_j,\bar{b}_j)\}$ where $b_j\leq \bar{b}_j$ and define $\bm{y}=(b_j)_{1\leq j\leq n}$. 
		Since $s_j\leq \bar{s}_j$ for all $j\leq n$, it follows that if $X_j=Y_j$ then $s_j=b_j$. 
		Since $d(X,Y)\leq t$, the Hamming distance between $\bm{x}$ and $\bm{y}$ satisfies $d(\bm{x},\bm{y})\leq d(X,Y)\leq t$. In the following, we will show that $\bm{x}$ belongs to an $(n,n-2t)$ GRS code over $\ff_p$ with minimum distance $2t+1$, and thus one can recover $\bm{x}$ from $\bm{y}$ using a decoder for the GRS code.
		From \eqref{eq:pc}, for $l=0,\ldots, 2t-1$ we have 
		\begin{align}
			0
			&=\sum_{j=1}^{n}\sum_{i=j}^{n}\omega_i\alpha_i^l c_j 
			\nonumber\\
			&=\sum_{i=1}^{n} \omega_i\alpha_i^l \sum_{j=1}^{i} c_j 
			\label{eq:s}\\
			&=\sum_{i=1}^{n} \omega_i\alpha_i^l {s}_i\nonumber,
		\end{align} where \eqref{eq:s} follows by noticing ${s}_i=\sum_{j=1}^{i}c_j$. 
		Therefore, $\bm{x}=(s_j)_{1\leq j\leq n}$ is a codeword of an $(n,n-2t)$ GRS code with minimum distance $2t+1$, and we can decode $\bm{y}$ to $\bm{x}$ since $d(\bm{x},\bm{y})\leq t$.
		Furthermore, we can recover $\bm{c}$ since $c_j=s_{j}-s_{j-1}$ for $j=1,\ldots,n$, where $s_0:=0$. 
	\end{proof}
	
	\begin{remark}
		The redundancy required to enable correction of $t$ errors for the subcode of $\cD_n\cap\{0,1\}^n$  that satisfies the weight condition $\wt(\bm{c}[j])\leq \wt(\cev{\bm{c}}[j])$ is at least $2t$.
		We note that if the numbers of the composition errors of each type are known a priori, one can use less redundancy for the code construction. This is because insertion errors and deletion errors in $M(\bm{c})$ amount to erasures in the sense that one can identify such errors by examining $(Y_j)_{1\leq j\leq n}$. More precisely, if $|Y_j|\neq 2$, then one knows $Y_j$ is in error (due to insertions or deletions) and may simply discard the information provided by $Y_j$. Therefore, if one knows a priori that there will be at most $t_1>0$ coordinates of $Y$ that are erasures (i.e., $|Y_j|\neq 2$) and at most $t-t_1$ substitution errors (i.e., $|Y_j|=2$ but $Y_j\neq X_j$), then the redundancy required for error correction can be reduced to be at least $2(t-t_1)+t_1=2t-t_1$ and decoding can be accomplished by using an error-erasure decoder for GRS codes.
	\end{remark}
	
	Although Lemma~\ref{le:basic} shows the binary subcode of $\cD_n$ that satisfies certain weight condition is capable of correcting composition errors, it is not obvious how to construct such a subcode efficiently. In the following, let us take a different approach to constructing codes derived from $\cD_n$ that fulfill the weight condition.
	
	\begin{construction}\label{con:ext}
		Let $n_1$ be a positive integer and assume $t< n_1/2$.
		Let $\cD_{n_1}=\{(c_j)_{1\leq j\leq n_1}\}$ be the linear code of length $n_1$ over $\ff_p$ as given in Construction~\ref{con:basic}, where $p\geq n_1+1$.
		For each $(c_j)_{1\leq j\leq n_1}\in\cD_{n_1}$, let $\bm{c}_j$ be the length-$(2p-1)$ binary string of $c_j$ ones followed by $2p-1-c_j$ zeros if $1\leq j\leq \lceil n_1/2\rceil$, and $\bm{c}_j$ be the length-$(2p-1)$ binary string of $p+c_j$ ones followed by $p-1-c_j$ zeros if $\lceil n_1/2\rceil +1\leq j\leq n_1$.
		Let $n=n_1(2p-1)$ and define $\cC\subset \{0,1\}^{n}$ to be
		\begin{align*}
			\cC=\{(\bm{c}_1,\ldots,\bm{c}_{n_1})\mid \bm{c}\in\cD_{n_1}\}.
		\end{align*}
	\end{construction}
	
	\begin{theorem}
		Let $\bm{v}\in\cC$ where $\cC$ is given in Construction~\ref{con:ext}. If $Y\in\fM_{n}$ is any erroneous prefix-suffix compositions resulting from at most $t$ errors in $M(\bm{v})$, then one can recover $\bm{v}$ from $Y$.
	\end{theorem}
	\begin{proof}
		As in the proof of Lemma~\ref{le:basic}, without loss of generality, we may assume $|Y_i|=2$ for $i=1,\ldots,n=n_1(2p-1)$. In fact,  to recover $\bm{v}$ from $Y$, we will only need to look at the multisets $Y_{2p-1},Y_{2(2p-1)},\ldots,Y_{n_1(2p-1)}$.
		
		Write $\bm{v}=(\bm{c}_1,\ldots,\bm{c}_{n_1})$. Note that by Construction~\ref{con:ext}, we have 
		$\wt(\bm{c}_j)<p$ for all $j\leq \lceil n_1/2\rceil$ and $\wt(\bm{c}_j)\geq p$ for all $j> \lceil n_1/2\rceil$. It follows that $\wt(\bm{v}[j(2p-1)])=\wt(\bm{c}_1,\ldots,\bm{c}_j)\leq \wt(\cev{\bm{v}}[j(2p-1)]) = \wt(\bm{c}_{n_1-j+1},\ldots,\bm{c}_{n_1})$ for $j\leq \lceil n_1/2\rceil$, and thus by Proposition~\ref{prop:weight} we have $\wt(\bm{v}[j(2p-1)])\leq \wt(\cev{\bm{v}}[j(2p-1)])$ for all $j=1,\ldots,n_1$.
		Let $s_j=\wt(\bm{v}[j(2p-1)])\bmod p$ and $\bar{s}_j=\wt(\cev{\bm{v}}[j(2p-1)])\bmod p$, and view $s_j,\bar{s}_j$ as elements in $\ff_p$.
		Observe that $s_j=\sum_{i=1}^{j}c_i$ for $j=1,\ldots,n$.
		
		Define $\bm{x}=(s_j)_{1\leq j\leq n_1}\in\ff_p^{n_1}$. 
		Moreover, let us write $Y_j=\{(a_j,b_j),(\bar{a}_j,\bar{b}_j)\}$ where $b_j\leq \bar{b}_j$ and define $\bm{y}=(b_j)_{1\leq j\leq n_1}\in\ff_p^{n_1}$. 
		Similarly to the proof of Lemma~\ref{le:basic}, one can show $d(\bm{x},\bm{y})\leq t$ and $\bm{x}$ belongs to an $(n_1,n_1-2t)$ GRS code over $\ff_p$ with minimum distance $2t+1$. Thus, one can recover $\bm{x}$ from $\bm{y}$. Below let us show how to reconstruct $\bm{v}$ from $\bm{x}=(s_1,\ldots,s_{n_1})$.
		
		By construction of $\cC$, in order to recover $\bm{v}$, it suffices to determine the weight of $\bm{c}_j,j=1,\ldots,n_1$. Now let us view $s_1,\ldots,s_{n_1}$ as integers. Since $s_1=\wt(\bm{v}[2p-1])\bmod p$ and $\wt(\bm{v}[2p-1])=\wt(\bm{c}_1)\leq p-1$, we have $\wt(\bm{c}_1)=s_1$. Moreover, for $j=2,\ldots,\lceil n_1/2\rceil$, since $0\leq \wt(\bm{v}[j(2p-1)])- \wt(\bm{v}[(j-1)(2p-1)])\leq p-1$, we have
		\begin{align*}
			\wt(\bm{c}_j)&=\wt(\bm{v}[j(2p-1)])- \wt(\bm{v}[(j-1)(2p-1)]) \\
			& = (s_j-s_{j-1}) \bmod p.
		\end{align*}
		At the same time, note that $\bar{s}_j=s_{n_1}-s_{n_1-j}$. Since $p\leq \wt(\cev{\bm{v}}[2p-1])=\wt(\bm{c}_{n_1})\leq 2p-1$ and $\bar{s}_1=\wt(\cev{\bm{v}}[2p-1])\bmod p$, we have $\wt(\bm{c}_{n_1})=p+\bar{s}_1$. Moreover, for $j=2,\ldots,\lfloor n_1/2\rfloor$, since $p\leq \wt(\cev{\bm{v}}[j(2p-1)])- \wt(\cev{\bm{v}}[(j-1)(2p-1)])\leq 2p-1$, we have
		\begin{align*}
			\wt(\bm{c}_{n_1-j+1})&=\wt(\cev{\bm{v}}[j(2p-1)])- \wt(\cev{\bm{v}}[(j-1)(2p-1)]) \\
			& = p+((\bar{s}_j-\bar{s}_{j-1}) \bmod p).
		\end{align*}
		Therefore, the weight of $\bm{c}_j,j=1,\ldots,n_1$ can be deduced from $s_1,\ldots,s_{n_1}$. It follows that $\bm{v}$ can be reconstructed.
	\end{proof}
	
	\begin{remark}
		The redundancy of the code $\cC$ in Construction~\ref{con:ext} is $O(n)$ while the number of errors it can correct is $t=O(\sqrt{n})$. In fact, the rate of $\cC$ is $o(1)$.
	\end{remark}

	A large number of bits the code $\cC$ in Construction~\ref{con:ext} are squandered on fulfilling the weight condition, leading to a mediocre trade-off between the redundancy and error-correcting capability. Inspired by the work of \cite{pattabiraman2023coding}, let us present below a construction of codes that have a systematic encoder and improved redundancy.

	\begin{construction}\label{con:encode}
		Let $n_1,n_2,t$ be positive integers such that $4(\log n_1 + 1)^2<n_2< n_1,2t<\sqrt{n_2}$ and $n_2+1$ is a power of $2$.
		Let $\ff_p$ be a finite field of order $p$, where $p\geq n_1+1$ is a prime, and let $\{\alpha_1,\dots,\alpha_{n_1}\}\subset\ff_p$.
		Let $H$ be a $2t\times n_1$ Vandermonde matrix given by 
		\begin{align*}
			H=
			\begin{pmatrix}
				1 & \dots & 1\\
				\alpha_1 & \dots & \alpha_{n_1}\\
				\vdots & \ddots & \vdots\\
				\alpha_1^{2t-1} & \dots & \alpha_{n_1}^{2t-1}
			\end{pmatrix}.			
		\end{align*}
		
		Let $\cS_{n_1}$ be the code in Lemma~\ref{le:ordered} and $\bm{w}=(w_j)_{1\leq j\leq n_1}\in\cS_{n_1}$. Let $\bm{x}=(x_j)_{1\leq j\leq n_1}$ where $x_j:=\wt(\bm{w}[j])$. Define $\bm{u}$ to be the length-$(2t\lceil\log p\rceil)$ string formed by the concatenation of the $2t$ binary representations of the symbols in $H\bm{x}^T\in\ff_p^{2t}$. 
		
		Let $\cB_{n_2}$ be the BCH code in Lemma~\ref{le:bch} of dimension $n_2-t\log(n_2+1)$ and designed distance $2t+1$. Let $\tilde{\bm{u}}$ be the length-$(n_2-t\log(n_2+1))$ string obtained from $\bm{u}$ by appending zeros to $\bm{u}$.\footnote{Note that by the choice of the parameters $n_1,n_2$ and $t$, we have $n_2-t\log(n_2+1)>2t\lceil\log p\rceil$ and $2t-1<2^{\lceil \log(n_2+1)/2 \rceil}+1$.  
		} 
		Let $G$ be a generator matrix of $\cB_{n_2}$ and let $\bm{v}=(v_j)_{1\leq j\leq n_2}=\tilde{\bm{u}}G$. Define $\bm{p}=(p_j)_{1\leq j\leq n_2}$ by
		\begin{align}
			p_j=\Big(v_j+\sum_{i=1}^{j-1}p_i \Big)\bmod 2.\label{eq:p}
		\end{align}
		Let $\bm{0}$ be the length-$n_2$ string of all zeros. Let $n=n_1+2n_2$ and define $\cC\subset \{0,1\}^n$ to be 
		\begin{align*}
			\cC=\{(\bm{0},\bm{w},\cev{\bm{p}})\mid \bm{w}\in\cS_{n_1}\}.
		\end{align*}
	\end{construction}
	
	\begin{theorem}\label{thm:encode}
		Let $\bm{c}\in\cC$ where $\cC$ is given in Construction~\ref{con:encode}. If $Y\in\fM_n$ is any erroneous prefix-suffix compositions resulting from at most $t$ errors in $M(\bm{c})$, then one can recover $\bm{c}$ from $Y$.
	\end{theorem}
	
	\begin{proof}
		Denote $X=M(\bm{c})$ and note that $\bm{c}=(\bm{0},\bm{w},\cev{\bm{p}})$. As in the proof of Lemma~\ref{le:basic}, without loss of generality, we may assume $|Y_j|=2$ for $j=1,\ldots,n_2$. Let us first show that $H\bm{x}^T$ can be recovered from $Y$.
		Define $\tilde{\bm{x}}=(\tilde{x}_j)_{1\leq j\leq n_2}$ where
		\begin{align*}
			\tilde{x}_j:=(\wt(\bm{c}[j])+\wt(\cev{\bm{c}}[j]))\bmod 2.
		\end{align*}
		Moreover, let us write $Y_j=\{(a_j,b_j),(\bar{a}_j,\bar{b}_j)\}$ where $b_j\leq \bar{b}_j$ for all $j=1,\ldots,n$, and define $\tilde{\bm{y}}=(\tilde{y}_j)_{1\leq j\leq n_2}$ where
		\begin{align*}
			\tilde{y}_j:=(b_j+\bar{b}_j) \bmod 2.
		\end{align*}
		Then for $j=1,\ldots,n_2$, we have
		\begin{align*}
			\tilde{x}_j\equiv\wt(\cev{\bm{c}}[j])\equiv\sum_{i=1}^{j}p_i \equiv v_j \pmod{2},
		\end{align*}
		where the first equivalence follows by the fact that the first $n_2$ bits of $\bm{c}$ are zeros and the last equivalence follows by \eqref{eq:p}. 
		Observe that $X_j=Y_j$ implies $\tilde{x}_j=\tilde{y}_j$.
		Since $d(X,Y)\leq t$, it follows that $d(\tilde{\bm{x}},\tilde{\bm{y}})\leq d(X,Y)\leq t$. By Lemma~\ref{le:bch}, $\tilde{\bm{x}}=\bm{v}$ is a codeword of the BCH code $\cB_{n_2}$ with minimum distance at least $2t+1$. Therefore, we can decode $\tilde{\bm{y}}$ to $\bm{v}$ and recover $H\bm{x}^T$. Note that by \eqref{eq:p} we can also recover $\cev{\bm{p}}$ from $\bm{v}$.
		It remains to show we can further recover $\bm{x}$, thereby recovering $\bm{w}$. 
		
		Let $\bm{y}=(b_j)_{n_2+1\leq n_2+j\leq n_2+n_1}$. 
		By Lemma~\ref{le:ordered}, we have $\wt(\bm{w}[j])\leq \wt(\cev{\bm{w}}[j])$ for all $j\leq n_1/2$, and thus by Proposition~\ref{prop:weight} we have $\wt(\bm{w}[j])\leq \wt(\cev{\bm{w}}[j])$ for all $j\leq n_1$. This implies the prefix weights of $\bm{w}$ can be distinguished from the suffix weights. It follows that $X_{n_2+j}=Y_{n_2+j}$ implies $x_j=b_j$ for $j=1,\ldots,n_1$, and therefore, $d(\bm{x},\bm{y})\leq d(X,Y)\leq t$. Given that $H\bm{y}^T-H\bm{x}^T$ is known, using a variant of the decoders for an $(n_1,n_1-2t)$ GRS code with evaluation points $\alpha_1,\ldots,\alpha_{n_1}$ and column multiplier $\omega_1=\ldots=\omega_{n_1}=1$, we can decode $\bm{y}$ to $\bm{x}$. Furthermore, we can recover $\bm{w}$ by setting $w_j=x_j-x_{j-1}$ for $j=1,\ldots,n_1$, where $x_0:=0$.
		Hence, if $d(X,Y)\leq t$, we are able to find $\bm{c}=(\bm{0},\bm{w},\cev{\bm{p}})$.
	\end{proof}

	\begin{remark}
		By Lemma~\ref{le:ordered}, the redundancy of $\cS_{n_1}$ is $O(\log n_1)$. Therefore, the redundancy of the code $\cC$ in Construction~\ref{con:encode} is $2n_2+O(\log n_1)=O(t\log n)$ assuming $n_2=O(\log^2 n_1)$ and $t=O(\log n_1)$. 
		Therefore, the code $\cC$ has rate asymptotically equal to $1$ while correcting $t=O(\log n)$ errors. 
		Furthermore, since the code $\cS_{n_1}$ can be encoded and decoded in polynomial time, it follows that $\cC$ can also be encoded and decoded with time complexity polynomial in $n$.
	\end{remark}

	\begin{remark}
		As one may have observed in the proof of Theorem~\ref{thm:encode}, the code consisting of strings of the form $(\bm{0},\cev{\bm{p}})$, where $\bm{0},\cev{\bm{p}}$ are as defined in Construction~\ref{con:encode}, actually can correct $t=O(\sqrt{n_2})$ composition errors. However, the redundancy is more than half of its length. Nevertheless, if $t$ is small enough, one may still use it as in Construction~\ref{con:encode} to construct larger codes that can correct $t$ errors. In fact, any binary linear code with minimum distance at least $2t+1$ would suffice for this purpose.
	\end{remark}
	
	The coding schemes presented above are based on the idea of enforcing the prefix weights form a codeword of certain error-correcting codes over $\ff_p$ where $p>2$ is a large enough prime. We next show that one can directly use binary error-correcting codes to form strings that can correct composition errors. Our starting point is the observation that $t$ errors in the prefix-suffix compositions of a string can induce at most $2t$ errors (including substitutions or erasures) in the string itself, if the prefix weights are distinguishable from the suffix weights.
	
	\begin{proposition}\label{prop:string}
		Let $\bm{c}\in\{0,1\}^n$ be such that $\wt(\bm{c}[j])\leq \wt(\cev{\bm{c}}[j])$ for all $j\leq n$ and let $Y\in\fM_n$ be the erroneous prefix-suffix compositions resulting from at most $t$ errors in $M(\bm{c})$, where $Y_j=\{(a_j,b_j),(\bar{a}_j,\bar{b}_j)\}$ with $b_j\leq \bar{b}_j$ for all $j\leq n$. Define $b_0=0$ and let $\bm{t}=t_1\ldots t_n\in\{0,1\}^n$ where $t_j:=(b_j-b_{j-1}) \bmod 2$ for $j=1,\ldots,n$. 
		Then $d(\bm{c},\bm{t})\leq 2t$.
	\end{proposition}
	
	\begin{proof}
		Let $s_j=\wt(\bm{c}[j])$ for $j\leq n$. Denote $\bm{x}=(s_j)_{1\leq j\leq n}$ and $\bm{y}=(b_j)_{1\leq j\leq n}$. Since $d(M(\bm{c}),Y)\leq t$ and $\wt(\bm{c}[j])\leq \wt(\cev{\bm{c}}[j]),b_j\leq \bar{b}_j$ for all $j\leq n$, we have $d(\bm{x},\bm{y})\leq t$.
		In other words, there are at most $t$ errors in $b_j,j=1,\ldots,n$. If either $b_j$ or $b_{j-1}$ is in error, then $d(t_j, c_j)\leq 1$. Therefore, if $b_j$ is in error, then $d(t_jt_{j+1},c_jc_{j+1})\leq 2$, i.e., an error in $\bm{y}$ can induce at most two errors in $\bm{t}$. In addition, if both $b_j$ and $b_{j-1}$ are correct, then $d(t_j,c_j)=0$. Thus, $d(\bm{t},\bm{c})\leq 2t$.
	\end{proof}
	
	In view of Proposition~\ref{prop:string}, if the string $\bm{c}$ belongs to a binary error-correcting code with minimum distance at least $4t+1$ and satisfies the weight condition that $\wt(\bm{c}[j])\leq \wt(\cev{\bm{c}}[j])$ for all $j\leq n$, then as long as $d(M(\bm{c}),Y)\leq t$ one is able to recover $\bm{c}$ from $Y$. In the following, we will construct $\bm{c}$ using the asymptotically good binary linear codes in Lemma~\ref{le:good}. To ensure the weight condition without sacrificing much of the rate, we will also rely on the codes in Lemma~\ref{le:ordered}. Our code construction is described as follows.
	
	\begin{construction}\label{con:good}
		Let $n_1,n_2,t$ be positive integers such that $n_1<n_2$ and $t<n_2/4$.
		Let $\cS_{n_1}$ be the code in Lemma~\ref{le:ordered} and $\bm{w}=(w_j)_{1\leq j\leq n_1}\in\cS_{n_1}$. Let 
		Let $\cA_{n_2}$ be the code in Lemma~\ref{le:good} with rate $R=n_1/n_2$ that can correct $2t$ errors and $G$ a generator matrix of $\cA_{n_2}$. Without loss of generality, assume $G$ is in the standard form where the first $n_1$ columns form the identity matrix of size $n_1$. Let $(\bm{w},\bm{p})=\bm{w}G$.
		Let $\bm{0}$ be the length-$(n_2-n_1)$ string of all zeros. Let $n=2n_2-n_1$ and define $\cC\subset \{0,1\}^n$ to be 
		\begin{align*}
			\cC=\{(\bm{0},\bm{w},{\bm{p}})\mid \bm{w}\in\cS_{n_1}\}.
		\end{align*}
	\end{construction}
	
	\begin{theorem}\label{thm:good}
		Let $\bm{c}\in\cC$ where $\cC$ is given in Construction~\ref{con:good}. If $Y\in\fM_n$ is any erroneous prefix-suffix compositions resulting from at most $t$ errors in $M(\bm{c})$, then one can recover $\bm{c}$ from $Y$.
	\end{theorem}
	
	\begin{proof}
		Without loss of generality, we may assume $|Y_j|=2$ for $j=1,\ldots,n$. Write $Y_j=\{(a_j,b_j),(\bar{a}_j,\bar{b}_j)\}$ with $b_j\leq \bar{b}_j$ for all $j\leq n$. Define $b_0=0$ and let $\bm{t}=t_1\ldots t_n\in\{0,1\}^n$ where $t_j:=(b_j-b_{j-1}) \bmod 2$ for $j=1,\ldots,n$.
		
		As a consequence of Lemma~\ref{le:ordered} and Proposition~\ref{prop:weight}, the codeword $\bm{c}=(\bm{0},\bm{w},{\bm{p}})\in\cC$ satisfies $\wt(\bm{c}[j])\leq \wt(\cev{\bm{c}}[j])$ for all $j\leq n$. Let $\tilde{\bm{t}}$ be the last $n_2$ bits of $\bm{t}$. Since $d(M(\bm{c}),Y)\leq t$, by Proposition~\ref{prop:string} we have $d((\bm{w},\bm{p}),\tilde{\bm{t}})\leq 2t$, where $(\bm{w},\bm{p})\in\cA_{n_2}$. It follows that one is able to recover $(\bm{w},\bm{p})$ and further obtain $\bm{c}$.
	\end{proof}
	
	\begin{remark}
		Assuming $n_1=\Theta(n_2)$, the code $\cC$ constructed in Construction~\ref{con:good} has redundancy $2(n_2-n_1)+O(\log n_1)=O(n)$ and can correct $t=O(n)$ errors in time polynomial in $n$.
	\end{remark}
	
	\subsection{Reconstruction of multiple strings}\label{sec:con-multiple}
	We next consider the problem of reconstructing multiple strings from the multiset union of their prefix-suffix compositions in the presence of composition errors. Let $\bm{c}_1,\ldots,\bm{c}_h\in\{0,1\}^n$ be $h>1$ binary strings. The prefix-suffix compositions of $\bm{c}_1,\ldots,\bm{c}_h$, denoted by $M(\bm{c}_1,\ldots,\bm{c}_h)$, is defined as the multiset union of the prefix-suffix compositions of the $h$ strings. Namely, $M(\bm{c}_1,\ldots,\bm{c}_h):=\bigcup_{i=1}^{h}M(\bm{c}_i)$. Note that we have $M(\bm{c}_1,\ldots,\bm{c}_h)\in\fM_n$. For any $Y\in\fM_n$, we say that $Y$ is a multiset of erroneous prefix-suffix compositions resulting from $t$ errors in $M(\bm{c}_1,\ldots,\bm{c}_h)$ if $d(M(\bm{c}_1,\ldots,\bm{c}_h),Y)=t$.
	
	Similar to the reconstruction of a single string, the task of reconstructing $\bm{c}_1,\ldots,\bm{c}_h$ from $Y$ is more manageable if the prefix and suffix weights of these strings are distinguishable in $M(\bm{c}_1,\ldots,\bm{c}_h)$. In the following, we present a method of jointly encoding $h$ arbitrary binary strings of length $k$ so that they can be reconstructed from their error-free prefix-suffix compositions, and then incorporate good error-correcting binary codes into this coding scheme to obtain codes that are capable of recovering $h$ strings from erroneous prefix-suffix compositions with at most $t$ errors. 
	
	Let $\bm{z}_1,\ldots,\bm{z}_h\in\{0,1\}^k$ be $h$ information strings of length $k$. We would like to jointly encode $\bm{z}_1,\ldots,\bm{z}_h$ so that the prefix and suffix weights of the coded strings are distinguishable in $M(\bm{z}_1,\ldots,\bm{z}_h)$.
	Our idea for realizing such weight conditions on the coded strings is to interleave $\bm{z}_1,\ldots,\bm{z}_h$ with ``short strings'' in a way similar to how we achieve the desired weight condition for a single string in Construction~\ref{con:ext}. For the ease of exposition, let us first assume the length $k$ of the information strings is even. Let $\{\bm{u}_{s,j}\in\{0,1\}^{h}\mid s=1,\ldots,2h;j=1,\ldots, k/2 \}$ be the ``short strings'' to be interleaved with $\bm{z}_1,\ldots,\bm{z}_h$. To construct the strings $\bm{u}_{s,j}$, we first divide each $\bm{z}_i,i=1,\ldots,h$ into two halves. More precisely, for $i=1,\ldots,h$ define
	\begin{align*}
		\bm{w}_{2i-1}&=\bm{z}_i[k/2],\\
		\bm{w}_{2i}&=\cev{\bm{z}_i}[k/2].
	\end{align*}
	Note that $\bm{w}_{2i-1}$ is the first $k/2$ bits of $\bm{i}$ and $\bm{w}_{2i}$ is the last $k/2$ bits of $\bm{z}_i$ in reverse order.
	Let us write $\bm{w}_s=(w_{s,1},\ldots,w_{s,k/2})$ for $s=1,\ldots,2h$. Roughly speaking, for $1\leq j\leq k/2$, the strings $\{\bm{u}_{s,j}\mid 1\leq s\leq 2h\}$ are constructed such that they record the indices $s$ for which $w_{s,j}<w_{s-1,j}$, i.e., $w_{s,j}=0$ and $w_{s-1,j}=1$.	
	More precisely, for $1\leq j\leq k/2$, let $\bm{r}_j=(r_{j,1},\ldots,r_{j,2h})\in\{0,1\}^{2h}$ be the indicator string of $(w_{0,j}:=0,w_{1,j},\ldots,w_{2h,j})$ where
	\begin{align}
		r_{j,s}:=\begin{cases}
			0 &\text{ if } w_{s,j}\geq w_{s-1,j},\\
			1 &\text{ if } w_{s,j} < w_{s-1,j},
		\end{cases}\label{eq:indicator}
	\end{align} where $s=1,\ldots,2h$. 
	Clearly, for each $j$, the number of indices $s$ for which $w_{s,j}<w_{s-1,j}$ is at most $h$, so we have $\wt(\bm{r}_j)\leq h$ for all $1\leq j\leq k/2$. Using the indicator string $\bm{r}_j$, the string $\bm{u}_{s,j}$ is then defined to be $ h -\wt(\bm{r}_j[s])$ zeros followed by $\wt(\bm{r}_j[s])$ ones:
	\begin{align}
		\bm{u}_{s,j}:=(\!\!\underbrace{0,\ldots,0}_{ h -\wt(\bm{r}_j[s])},\,\overbrace{1,\ldots,1}^{\wt(\bm{r}_j[s])}\,).\label{eq:short-string}
	\end{align}
	Below is a simple example to illustrate the construction of the strings $\bm{u}_{s,j}$.
	
	\begin{example}
		Assume $h=2,k=6$. Consider the strings $\bm{z}_1,\bm{z}_2$ given by 
		\begin{align*}
			\begin{pmatrix}
				\bm{z}_1\\
				\bm{z}_2
			\end{pmatrix}
			=
			\begin{pmatrix}
				1 & 0 & 1 & 0 & 1 & 0\\
				0 & 0 & 1 & 0 & 1 & 1 
			\end{pmatrix}.
		\end{align*}
		Then strings $\bm{w}_1,\bm{w}_2,\bm{w}_3,\bm{w}_4$ are
		\begin{align*}
			\begin{pmatrix}
				\bm{w}_1\\
				\bm{w}_2\\
				\bm{w}_3\\
				\bm{w}_4
			\end{pmatrix}
			=
			\begin{pmatrix}
				1 & 0 & 1 \\
				0 & 1 & 0 \\
				0 & 0 & 1 \\
				1 & 1 & 0 
			\end{pmatrix}
		\end{align*}
		The strings $\bm{r}_1,\bm{r}_2,\bm{r}_3$ can be found to be 
		\begin{align*}
			\begin{pmatrix}
				\bm{r}_1\\
				\bm{r}_2\\
				\bm{r}_3
			\end{pmatrix}
			=
			\begin{pmatrix}
				0 & 1 & 0 & 0\\
				0 & 0 & 1 & 0\\
				0 & 1 & 0 & 1 
			\end{pmatrix}.
		\end{align*}
		So the strings $\bm{u}_{s,j}$ where $1\leq s\leq 4,1\leq j\leq 3$ can be determined to be
		\begin{align*}
			\begin{pmatrix}
				\bm{u}_{1,1}\\
				\bm{u}_{2,1}\\
				\bm{u}_{3,1}\\
				\bm{u}_{4,1}
			\end{pmatrix}
			=
			\begin{pmatrix}
				0 & 0 \\
				0 & 1 \\
				0 & 1 \\
				0 & 1 
			\end{pmatrix},\quad 
			\begin{pmatrix}
				\bm{u}_{1,2}\\
				\bm{u}_{2,2}\\
				\bm{u}_{3,2}\\
				\bm{u}_{4,2}
			\end{pmatrix}
			=
			\begin{pmatrix}
				0 & 0 \\
				0 & 0 \\
				0 & 1 \\
				0 & 1 
			\end{pmatrix},\quad 
			\begin{pmatrix}
				\bm{u}_{1,3}\\
				\bm{u}_{2,3}\\
				\bm{u}_{3,3}\\
				\bm{u}_{4,3}
			\end{pmatrix}
			=
			\begin{pmatrix}
				0 & 0 \\
				0 & 1 \\
				0 & 1 \\
				1 & 1 
			\end{pmatrix}.
		\end{align*}
	\end{example}
	
	Next, the strings $\bm{w}_s$ are interleaved with the strings $\bm{u}_{s,j}$ to form strings
	\begin{align*}
		\bm{v}_{2i-1}&:=(\bm{u}_{2i-1,1},w_{2i-1,1},\ldots,\bm{u}_{2i-1, k/2},w_{2i-1, k/2}),\\
		\bm{v}_{2i}&:=(\bm{u}_{2i,1},w_{2i,1},\ldots,\bm{u}_{2i, k/2},w_{2i, k/2}).
	\end{align*}
	The encoded string corresponding to $\bm{z}_i$ is formed by appending the reversal of $\bm{v}_{2i}$ to $\bm{v}_{2i-1}$, defined to be 
	\begin{align}
		\bm{c}_i:=(\bm{v}_{2i-1},\cev{\bm{v}_{2i}}),\label{eq:multi-free}
	\end{align} 
	for $i=1,\ldots,h$. 
	The next lemma shows that the weights of $\bm{c}_1[l],\cev{\bm{c}_1}[l],\ldots,\bm{c}_h[l],\cev{\bm{c}_h}[l]$ have the same order for all $l=1,\ldots,k(h+1)/2$.
	
	\begin{lemma}\label{le:h-weight-even-lenght}
	Assume $k$ is even. Let $\bm{c}_1,\ldots,\bm{c}_h\in \{0,1\}^{k(h+1)}$ be constructed as in \eqref{eq:multi-free}. Then for any $1\leq l\leq k(h+1)/2$, it holds that
	\begin{align*}
		\wt(\bm{c}_1[l])\leq\wt(\cev{\bm{c}_1}[l]) 
		\leq \ldots \leq \wt(\bm{c}_h[l]) \leq \wt(\cev{\bm{c}_h}[l]).
	\end{align*}
	\end{lemma}
	
	\begin{proof}
		By \eqref{eq:multi-free}, it suffices to show that 
		\begin{align*}
				\wt(\bm{v}_1[l])\leq\wt({\bm{v}_2}[l]) \leq \ldots \leq \wt({\bm{v}_{2h}}[l])
		\end{align*} holds for $1\leq l\leq k(h+1)/2$.
		Consider $j \in \{ 1,\ldots, k/2 \}$. Note that by \eqref{eq:short-string} we have $\wt(\bm{u}_{s,j})=\wt(\bm{r}_j[s])$ for $s=1,\ldots,2h$. Therefore, $\wt(\bm{u}_{1,j})\leq\ldots \leq \wt(\bm{u}_{2h,j})$. Furthermore, since $\wt(\bm{u}_{s,j})\leq \wt(\bm{u}_{s+1,j})$ and the zeros in each $\bm{u}_{s,j}$ always come first, the weight of each prefix of $\bm{u}_{s,j}$ is no larger than the weight of the prefix of $\bm{u}_{s+1,j}$ of the same length. Namely, for $l=1,\ldots,h$ we have
		\begin{align}
			 \wt(\bm{u}_{1,j}[l])\leq \wt(\bm{u}_{2,j}[l])\leq \ldots\leq \wt(\bm{u}_{2h,j}[l]). \label{eq:weight-part1}
		\end{align}
		From \eqref{eq:indicator} we have if $w_{s,j}=0$ and $w_{s-1,j}=1$ then $\wt(\bm{u}_{s,j})=\wt(\bm{u}_{s-1,j})+1$; otherwise $\wt(\bm{u}_{i,j})=\wt(\bm{u}_{i-1,j})$. Thus, we obtain $\wt(\bm{u}_{s-1,j}w_{s-1,j})\leq \wt(\bm{u}_{s,j},w_{s,j})$, and therefore, 
		\begin{align}
			\wt(\bm{u}_{1,j},w_{1,j})\leq\wt(\bm{u}_{2,j},w_{2,j})\leq \ldots \leq \wt(\bm{u}_{2h,j},w_{2h,j}).\label{eq:weight-part2}
		\end{align}
		Since both \eqref{eq:weight-part1} and \eqref{eq:weight-part2} hold for all $j=1,\ldots,k/2$, we have
		\[
		\wt(\bm{v}_1[l])\leq\wt({\bm{v}_2}[l]) \leq \ldots \leq \wt({\bm{v}_{2h}}[l]),\quad l=1,\ldots,k(h+1)/2,
		\]
		or equivalently,
		\begin{align*}
			\wt(\bm{c}_1[l])\leq\wt({\cev{\bm{c}_1}}[l]) \leq 
			\ldots \leq \wt(\bm{c}_h[l])\leq\wt({\cev{\bm{c}_h}}[l]),\quad l=1,\ldots,k(h+1)/2.
		\end{align*}
	\end{proof}

			Let $\phi_k\colon\{0,1\}^{kh}\to \{0,1\}^{kh(h+1)}$ be the mapping that maps $(\bm{z}_1,\ldots,\bm{z}_h)\in\{0,1\}^{kh}$ to $(\bm{c}_1,\ldots,\bm{c}_h)\in \{0,1\}^{kh(h+1)}$ where $\bm{c}_1,\ldots,\bm{c}_h$ are constructed as in \eqref{eq:multi-free}.
		The above assumes the length $k$ of the information strings is even but the case where $k$ is odd can also be handled without much difficulty. Now let us assume $k$ is odd. In this case, for each $\bm{z}_i:=(z_{i,1},\ldots,z_{i,k})$, define \[
		\bar{\bm{z}}_i=(z_{i,1},\ldots,z_{i,\lceil k/2 \rceil},z_{i,\lceil k/2 \rceil},z_{i,\lceil k/2 \rceil+1}\ldots,z_{i,k})\in\{0,1\}^{k+1}.
		\]
		In other words, $\bar{\bm{z}}_i$ is obtained from $\bm{z}_i$ by repeating its $\lceil k/2 \rceil$-th bit. Let $(\bar{\bm{c}}_1,\ldots,\bar{\bm{c}}_h)=\phi_{k+1}(\bar{\bm{z}}_1,\ldots,\bar{\bm{z}}_h)$. The coded string corresponding to $\bm{z}_i$ is defined to the string $\bm{c}_i$ obtained from $\bar{\bm{c}}_i$ by deleting its bits with coordinates in  $\{ \lceil k/2 \rceil (h+1)+1,\ldots,\lceil k/2 \rceil(h+1)+h+1\}$. Namely,
		\begin{align}
			\bm{c}_i:=(\bar{c}_{i,1},\ldots,\bar{c}_{i,\lceil k/2 \rceil(h+1)},\bar{c}_{i,(\lceil k/2 \rceil+1)(h+1)+1},\ldots,\bar{c}_{i,(k+1)(h+1)})\in \{0,1\}^{k(h+1)},\label{eq:multi-free-odd}
		\end{align} where $(\bar{c}_{i,1},\ldots,\bar{c}_{i,(k+1)(h+1)})=\bar{\bm{c}}_i$. So for odd $k$ we may define the mapping $\phi_k\colon\{0,1\}^{kh}\to \{0,1\}^{kh(h+1)}$ by $(\bm{z}_1,\ldots,\bm{z}_h)\mapsto (\bm{c}_1,\ldots,\bm{c}_h)$ where $\bm{c}_1,\ldots,\bm{c}_h$ are constructed as in \eqref{eq:multi-free-odd}. It is easy to check that Lemma~\ref{le:h-weight-even-lenght} also holds for $\bm{c}_1,\ldots,\bm{c}_h$ in the case of odd $k$ by an inspection of the weights of the prefixes and suffixes of $\bar{\bm{c}}_1,\ldots,\bar{\bm{c}}_h$ of length up to $(k+1)(h+1)/2= \lceil k/2 \rceil(h+1)\geq \lceil k(h+1)/2\rceil $, as in the proof of Lemma~\ref{le:h-weight-even-lenght}. Therefore, we have the following lemma.
		\begin{lemma}\label{le:h-weight}
			Assume $k\geq 1$. Let $(\bm{c}_1,\ldots,\bm{c}_h)=\phi_k(\bm{z}_1,\ldots,\bm{z}_h)$ where $\bm{z}_i\in\{0,1\}^k$ for all $i=1,\ldots,h$. Then for any $1\leq l\leq \lceil k(h+1)/2 \rceil$, it holds that
			\begin{align*}
				\wt(\bm{c}_1[l])\leq\wt(\cev{\bm{c}_1}[l])
				\leq \ldots \leq \wt(\bm{c}_h[l]) \leq \wt(\cev{\bm{c}_h}[l]).
			\end{align*}
		\end{lemma}
		
		As a consequence of Lemma~\ref{le:h-weight}, the strings $(\bm{c}_1,\ldots,\bm{c}_h)=\phi_k(\bm{z}_1,\ldots,\bm{z}_h)$ can be reconstructed from $M(\bm{c}_1,\ldots,\bm{c}_h)$ and so are the strings $\bm{z}_1,\ldots,\bm{z}_h$.
		
		\begin{theorem}\label{thm:multi-error-free}
			Assume $k\geq 1$. Let $(\bm{c}_1,\ldots,\bm{c}_h)=\phi_k(\bm{z}_1,\ldots,\bm{z}_h)$ where $\bm{z}_i\in\{0,1\}^k$ for all $i=1,\ldots,h$. Then one can recover $\bm{z}_1,\ldots,\bm{z}_h$ from $M(\bm{c}_1,\ldots,\bm{c}_h)$.
		\end{theorem}
		
		\begin{proof}
			Let $n=k(h+1)$ and denote $X=M(\bm{c}_1,\ldots,\bm{c}_h)$ where $X=(X_l)_{1\leq l\leq n}$. Note that $X_l$ is a multiset of size $2h$. Let us write $X_l=\{(a_{l,1},b_{l,1}),(\bar{a}_{l,1},\bar{b}_{l,1})\ldots,(a_{l,h},b_{l,h}),(\bar{a}_{l,h},\bar{b}_{l,h})\}$ where $b_{l,1}\leq\bar{b}_{l,1} \ldots \leq b_{l,h}\leq \bar{b}_{l,h}$. By Lemma~\ref{le:h-weight}, for any $1\leq l\leq \lceil n/2\rceil$, it holds that
			\begin{align*}
				\wt(\bm{c}_1[l])\leq\wt(\cev{\bm{c}_1}[l])
				\leq \ldots \leq \wt(\bm{c}_h[l]) \leq \wt(\cev{\bm{c}_h}[l]).
			\end{align*}
			Therefore, for any $1\leq l\leq\lceil n/2\rceil$, we have $b_{l,i}=\wt(\bm{c}_i[l])$ and $\bar{b}_{l,i}=\wt(\cev{\bm{c}_i}[l])$ where $i=1,\ldots,h$.
			Define $b_{0,i}=\bar{b}_{0,i}=0$. Let $\bm{c}_i=(c_{i,1},\ldots,c_{i,n})$. Then we have $c_{i,l}=b_{l,i}-b_{l-1,i}$ for $l=1,\ldots,\lceil n/2\rceil$ and $c_{i,n-l+1}=\bar{b}_{l,i}-\bar{b}_{l-1,i}$ for $l=1,\ldots,\lfloor n/2\rfloor$. Therefore, we can recover $\bm{c}_1,\ldots,\bm{c}_h$ from $M(\bm{c}_1,\ldots,\bm{c}_h)$. Moreover, by construction of $\bm{c}_i$, we have $\bm{z}_i=(c_{i,h+1},c_{i,2(h+1)},\ldots,c_{i,k(h+1)})$ for all $i=1,\ldots,h$.
		\end{proof}
	
		\begin{remark}
			The rate of the code $\{\phi_k(\bm{z}_1,\ldots,\bm{z}_h) \mid  \bm{z}_i\in\{0,1\}^k; i=1,\ldots,h\}$ is $1/(h+1)$. As a comparison, the codes presented in \cite{gabrys2022reconstruction} have rate asymptotically equal to $1/h$ but can only afford reconstruction of distinct strings. In contrast, our construction here allows for reconstruction of arbitrary strings. Moreover, the strings can be reconstructed in time linear in $n$, assuming $h$ is a constant independent of $n$. 
		\end{remark}
		
		Now if we assume the string $\bm{z}_1,\ldots,\bm{z}_h$ are not arbitrary strings but codewords that belong to binary codes with good error-correcting capability, then reconstruction of $\bm{z}_1,\ldots,\bm{z}_h$ from erroneous prefix-suffix compositions of $(\bm{c}_1,\ldots,\bm{c}_h)=\phi_k(\bm{z}_1,\ldots,\bm{z}_h)$ would be straightforward. This is because one can identify erroneous versions of the strings $\bm{c}_1,\ldots,\bm{c}_h$ from the erroneous compositions, and the structure of these erroneous strings lends itself for the reconstruction of $\bm{z}_1,\ldots,\bm{z}_h$.  
		\begin{theorem}
			Assume $k\geq 1$. Let $(\bm{c}_1,\ldots,\bm{c}_h)=\phi_k(\bm{z}_1,\ldots,\bm{z}_h)$ where $\bm{z}_1,\ldots,\bm{z}_h\in\cA_k$ with $\cA_k$ being the code of length $k$ in Lemma~\ref{le:good} that is able to correct $4t$ errors. If $Y\in\fM_n$ is any erroneous prefix-suffix compositions resulting from at most $t$ errors in $M(\bm{c}_1,\ldots,\bm{c}_h)$, then one can recover $\bm{z}_1,\ldots,\bm{z}_h$ from $Y$.
		\end{theorem}
		\begin{proof}
			Denote $n=k(h+1)$. Without loss of generality, we may assume $|Y_l|=2h$ for all $l=1,\ldots,n$. Let us write $Y_l=\{(a_{l,1},b_{l,1}),(\bar{a}_{l,1},\bar{b}_{l,1})\ldots,(a_{l,h},b_{l,h}),(\bar{a}_{l,h},\bar{b}_{l,h})\}$ where $b_{l,1}\leq\bar{b}_{l,1} \ldots \leq b_{l,h}\leq \bar{b}_{l,h}$.
			Define $b_{0,i}=\bar{b}_{0,i}=0$ for $i=1,\ldots,h$. Moreover, for $i=1,\ldots,h$ let $t_{i,l}:=(b_{l,i}-b_{l-1,i})\bmod 2$ where $l=1,\ldots,\lceil n/2\rceil$ and $t_{i,n-l+1}:=(\bar{b}_{l,i}-\bar{b}_{l-1,i})\bmod 2$ where $l=1,\ldots,\lfloor n/2\rfloor$. Define $\bm{t}_i=(t_{i,h+1},t_{i,2(h+1)},\ldots,t_{i,k(h+1)})$. Observe that if $\max\{d(\bm{t}_i,\bm{z}_i)\mid 1\leq i\leq h\}\leq 4t$, then we are able to recover $\bm{z}_1,\ldots,\bm{z}_h\in\cA_k$ since $\cA_k$ can correct $4t$ errors. So it remains to bound the Hamming distance between $\bm{t}_i$ and $\bm{z}_i$.
			
			 By Lemma~\ref{le:h-weight}, for any $1\leq l\leq \lceil n/2\rceil$, it holds that
			\begin{align*}
				\wt(\bm{c}_1[l])\leq\wt(\cev{\bm{c}_1}[l])
				\leq \ldots \leq \wt(\bm{c}_h[l]) \leq \wt(\cev{\bm{c}_h}[l]).
			\end{align*}
			Since $d(M(\bm{c}_1,\ldots,\bm{c}_h),Y)\leq t$, the Hamming distance between $(\wt(\bm{c}_i[l]))_{1\leq l\leq \lceil n/2\rceil}$ and $(b_{l,i})_{1\leq l\leq \lceil n/2\rceil }$ is at most $t$ for all $i$. Similarly, the Hamming distance between $(\wt(\cev{\bm{c}_i}[l]))_{1\leq l\leq \lfloor n/2\rfloor}$ and $(\bar{b}_{l,i})_{1\leq l\leq \lfloor n/2\rfloor }$ is at most $t$ for all $i$. 
			Therefore, there are at most $2t$ errors in $(b_{l,i})_{1\leq l\leq n}$. Similar to the observation in Proposition~\ref{prop:string}, since each bit of $\bm{t}_i$ is the difference of two consecutive bits of $(b_{l,i})_{1\leq l\leq n}$, there are at most $4t$ errors in $\bm{t}_i$. Namely, $d(\bm{t}_i,\bm{z}_i)\leq 4t$ for all $i$.
		\end{proof}
		
		\begin{remark}
			Assuming $h$ is a constant independent of $k$, the code $\{\phi_k(\bm{z}_1,\ldots,\bm{z}_h) \mid  \bm{z}_i\in\cA_k; i=1,\ldots,h\}$ has constant rate and can correct a constant fraction of composition errors relative to $n$.
		\end{remark}
		
		\section{Concluding remarks}\label{sec:re}
		
		We have presented several efficiently constructable and decodable codes that allow reconstruction of a single string from its erroneous prefix-suffix compositions. The trade-off between the amount of redundancy and the maximum number of composition errors that can be corrected for each coding scheme is summarized in Table~\ref{tab:comp}.
		
		In addition, we have also presented coding methods for reconstructing $h$ strings from their error-free or erroneous prefix-suffix compositions. For the error-free case, the rate of our constructed codes equals $1/(h+1)$. As for the erroneous case, the corresponding codes have constant rate and can correct a constant fraction of composition errors relative to the string length $n$. In both cases, our codes are efficiently constructable and decodable. 
		
		We note that there are a few directions that are interesting for further investigation. For the case of reconstructing $h$ strings, it is possibly necessary for the code schemes to have $h$-fold coding overhead if one insists on reconstructing \emph{any} $h$ information strings. However, it may be possible to sacrifice a negligibly small set of information strings, allowing them to fail in the reconstruction, in order to reduce the coding overhead. Moreover, the strings we used for interleaving are of the same length, which is more than what is needed for reconstruction of specific strings. Therefore, it may be worth looking into variable-length codes to see if they can have reduced redundancy. Separately, our definition of a single composition error may be too strong in the sense that if a single composition is in error then \emph{all} compositions of the same size as it are assumed to be in error, and thus their information is simply disregarded. However, it is natural to expect a closer inspection of the admissible patterns of composition errors will lead to coding schemes with improved error-correcting capability as well as reduced redundancy.
		
		\begin{table}[ht]
			\renewcommand{\arraystretch}{1.3}
			\centering
			\begin{tabular}{|c|c|c|}
				\hline
				Coding scheme  & Redundancy & Maximum correctable errors\\
				\hline\hline
				Construction~\ref{con:ext} & $O(n)$ & $t=\Theta(\sqrt{n})$ \\
				\hline
				Construction~\ref{con:encode} & $O(\log^2 n)$ & $t=\Theta(\log n)$ \\
				\hline
				Construction~\ref{con:good} & $O(n)$ & $t=\Theta(n)$ \\
				\hline
			\end{tabular}
			\caption{A comparison of the length-$n$ codes that allow reconstruction of a single string in this work.}
			\label{tab:comp}
		\end{table}

\bibliographystyle{IEEEtran}
\bibliography{CompositionErrors}

\begin{thebibliography}{10}
\providecommand{\url}[1]{#1}
\csname url@samestyle\endcsname
\providecommand{\newblock}{\relax}
\providecommand{\bibinfo}[2]{#2}
\providecommand{\BIBentrySTDinterwordspacing}{\spaceskip=0pt\relax}
\providecommand{\BIBentryALTinterwordstretchfactor}{4}
\providecommand{\BIBentryALTinterwordspacing}{\spaceskip=\fontdimen2\font plus
\BIBentryALTinterwordstretchfactor\fontdimen3\font minus
  \fontdimen4\font\relax}
\providecommand{\BIBforeignlanguage}[2]{{%
\expandafter\ifx\csname l@#1\endcsname\relax
\typeout{** WARNING: IEEEtran.bst: No hyphenation pattern has been}%
\typeout{** loaded for the language `#1'. Using the pattern for}%
\typeout{** the default language instead.}%
\else
\language=\csname l@#1\endcsname
\fi
#2}}
\providecommand{\BIBdecl}{\relax}
\BIBdecl

\bibitem{al2017mass}
A.~Al~Ouahabi, J.-A. Amalian, L.~Charles, and J.-F. Lutz, ``Mass spectrometry
  sequencing of long digital polymers facilitated by programmed inter-byte
  fragmentation,'' \emph{Nature communications}, vol.~8, no.~1, p. 967, 2017.

\bibitem{launay2021precise}
K.~Launay, J.-a. Amalian, E.~Laurent, L.~Oswald, A.~Al~Ouahabi, A.~Burel,
  F.~Dufour, C.~Carapito, J.-l. Cl{\'e}ment, J.-F. Lutz \emph{et~al.},
  ``Precise alkoxyamine design to enable automated tandem mass spectrometry
  sequencing of digital poly (phosphodiester) s,'' \emph{Angewandte Chemie},
  vol. 133, no.~2, pp. 930--939, 2021.

\bibitem{chen2020bioinformatics}
C.~Chen, J.~Hou, J.~J. Tanner, and J.~Cheng, ``Bioinformatics methods for mass
  spectrometry-based proteomics data analysis,'' \emph{International journal of
  molecular sciences}, vol.~21, no.~8, p. 2873, 2020.

\bibitem{acharya2015string}
J.~Acharya, H.~Das, O.~Milenkovic, A.~Orlitsky, and S.~Pan, ``String
  reconstruction from substring compositions,'' \emph{SIAM Journal on Discrete
  Mathematics}, vol.~29, no.~3, pp. 1340--1371, 2015.

\bibitem{pattabiraman2023coding}
S.~Pattabiraman, R.~Gabrys, and O.~Milenkovic, ``Coding for polymer-based data
  storage,'' \emph{IEEE Transactions on Information Theory}, vol.~69, no.~8,
  pp. 4812--4836, 2023.

\bibitem{banerjee2023insertion}
A.~Banerjee, A.~Wachter-Zeh, and E.~Yaakobi, ``Insertion and deletion
  correction in polymer-based data storage,'' \emph{IEEE Transactions on
  Information Theory}, vol.~69, no.~7, pp. 4384--4406, 2023.

\bibitem{gupta2025new}
U.~Gupta and H.~Mahdavifar, ``A new algebraic approach for string
  reconstruction from substring compositions,'' \emph{IEEE Transactions on
  Information Theory}, vol.~71, no.~1, pp. 125--137, 2025.

\bibitem{gabrys2022reconstruction}
R.~Gabrys, S.~Pattabiraman, and O.~Milenkovic, ``Reconstruction of sets of
  strings from prefix/suffix compositions,'' \emph{IEEE Transactions on
  Communications}, vol.~71, no.~1, pp. 3--12, 2023.

\bibitem{ye2024reconstruction}
Z.~Ye and O.~Elishco, ``Reconstruction of a single string from a part of its
  composition multiset,'' \emph{IEEE Transactions on Information Theory},
  vol.~70, no.~6, pp. 3922--3940, 2024.

\bibitem{yang2025reconstruction}
Y.~Yang and Z.~Chen, ``Reconstruction of multiple strings of constant weight
  from prefix--suffix compositions,'' \emph{Entropy}, vol.~27, no.~1, p.~39,
  2025.

\bibitem{sima2023constant}
J.~Sima, Y.-H. Li, I.~Shomorony, and O.~Milenkovic, ``On constant-weight binary
  {$B_2$}-sequences,'' in \emph{2023 IEEE International Symposium on
  Information Theory (ISIT)}.\hskip 1em plus 0.5em minus 0.4em\relax IEEE,
  2023, pp. 886--891.

\bibitem{macwilliams1977theory}
F.~J. MacWilliams and N.~J.~A. Sloane, \emph{The theory of error-correcting
  codes}.\hskip 1em plus 0.5em minus 0.4em\relax Elsevier, 1977, vol.~16.

\bibitem{sipser1996expander}
M.~Sipser and D.~A. Spielman, ``Expander codes,'' \emph{IEEE transactions on
  Information Theory}, vol.~42, no.~6, pp. 1710--1722, 1996.

\bibitem{forney1965concatenated}
G.~D. Forney, \emph{Concatenated codes.}\hskip 1em plus 0.5em minus 0.4em\relax
  Cambridge, MA: MIT Press, 1965.

\bibitem{zemor2001expander}
G.~Z{\'e}mor, ``On expander codes,'' \emph{IEEE Transactions on Information
  Theory}, vol.~47, no.~2, pp. 835--837, 2001.

\bibitem{forney1966generalized}
G.~Forney, ``Generalized minimum distance decoding,'' \emph{IEEE Transactions
  on Information Theory}, vol.~12, no.~2, pp. 125--131, 1966.

\end{thebibliography}

%\begin{IEEEbiographynophoto}{Zitan Chen} received the Ph.D. degree in electrical engineering from the University of Maryland, College Park, MD, USA, in 2020. He is currently an Assistant Professor with the School of Science and Engineering, The Chinese University of Hong Kong, Shenzhen.
\end{document}